\documentclass[11pt]{article}
\usepackage{parskip}
\usepackage{CJK}
\usepackage{amsmath,amsthm,amssymb}
\usepackage{graphicx}
\usepackage[dvipsnames]{xcolor}
\usepackage[hidelinks]{hyperref}
\definecolor{MaterialsCoral}{cmyk}{0, 0.75, 0.5, 0}
\definecolor{MaterialsSky}{cmyk}{0.6, 0, 0, 0}
\definecolor{MaterialsSun}{cmyk}{0, 0.2, 0.6, 0.05}
\definecolor{MaterialsGrass}{cmyk}{0.65, 0, 0.3, 0}

\usepackage{doi}
\usepackage{fullpage,palatino}
\usepackage{enumitem}
\setlist{noitemsep,leftmargin=*}
\newtheorem{theorem}{Theorem}
\newtheorem{lemma}{Lemma}
\newcommand{\Xomit}[1]{}
\newcommand{\leaf}[1]{L(#1)}
\newcommand{\algo}{\text{tree}(G)}
\newcommand{\rank}[1]{r(#1)}

\title{A Simple $2$-Approximation for Maximum-Leaf Spanning
  Tree\thanks{See~\cite{Liao01} for a preliminary version of this work, which was conducted independently of~\cite{Solis-Oba98,Solis-ObaBL17}.}}

\author{I-Cheng Liao\thanks{Department of Computer Science
and Information Engineering, National Chung-Cheng University, Chia-Yi, Taiwan.}
\and
  Hsueh-I Lu\thanks{Corresponding author. Email:
    hil@csie.ntu.edu.tw. Department of Computer Science and
    Information Engineering, National Taiwan University, Taipei, Taiwan. Research of this author is supported by
Grant 110--2221--E--002--075--MY3 of National Science and Technology Council (NSTC), formerly known as Ministry of Science and Technology (MOST).}}

\begin{document}
\begin{CJK*}{UTF8}{bkai}
\maketitle
\begin{abstract}
For an $m$-edge connected simple graph $G$, finding a spanning tree of
$G$ with the maximum number of leaves is MAXSNP-complete.  The problem
remains NP-complete even if $G$ is planar and the maximal degree of
$G$ is at most four.  Lu and Ravi gave the first known polynomial-time
approximation algorithms with approximation factors $5$ and $3$.
Later, they obtained a $3$-approximation algorithm that runs in near-linear
time.  The best known result is Solis-Oba, Bonsma, and Lowski's
$O(m)$-time $2$-approximation algorithm.  We show an alternative
simple $O(m)$-time $2$-approximation algorithm whose analysis is
simpler. 
This paper is dedicated to the cherished memory of our dear friend, Professor Takao Nishizeki.
\end{abstract}

\section{Introduction}
\label{section:section1}
We dedicate this paper to honoring the enduring contributions of our beloved friend,
Professor Takao Nishizeki, to the field of graph algorithms. Throughout his distinguished career, Professor Nishizeki contributes many classic results to the field, including interesting work on approximation algorithms~\cite{ChibaNS82,ItoDZN08,NishizekiAW83,ObataN15} and spanning trees~\cite{KawabataN14,MiuraAN05,MiuraTNN99,TakamizawaNS80}. We remember him with deep affection and gratitude for his remarkable achievements.

For an $n$-vertex $m$-edge connected simple graph $G$, finding a
spanning tree of $G$ with the maximum number of leaves is NP-complete
even if $G$ is planar and has maximum degree at most
four~\cite{GareyJ79}.  Moreover, the problem is
MAXSNP-complete~\cite{GalbiatiMM94}, implying that it does not admit a
polynomial-time approximation scheme unless
$\text{NP}=\text{P}$~\cite{AroraLMSS98,
AroraS98,PapaYanna88}.  The
problem finds applications in communication networks, circuit
layouts, and computer graphics~\cite{Diaz-GutierrezBGP06,Dijkstra74,Tchuente81}.
See~\cite{Rosamond2016} for an up-to-date survey on the maximum leaf spanning tree problem.

Lu and Ravi~\cite{LuR92} gave
the first polynomial-time approximation algorithms for the problem
based on local search. Their $3$-approximation (respectively,
$5$-approximation) algorithm runs in $O(m^2n^3)$ (respectively,
$O(mn^2)$) time.  They later gave a new $3$-approximation algorithm
that runs in $O(m\cdot\alpha(m,n))$ time~\cite{LuRavi96}.  Solis-Oba,
Bonsma, and Lowski~\cite{Solis-Oba98,Solis-ObaBL17} improved their
result by giving an $O(m)$-time $2$-approximation algorithm.
Their
algorithm has multiple phases of loops.  Their proof for the
approximation ratio is also somewhat involved.  We give an alternative
$O(m)$-time $2$-approximation algorithm that is simpler:
(1) Our
algorithm has only one simple loop, easily implementable to
run in linear time using basic data structures like arrays and linked
lists.  (2) Our proof for the approximation ratio is shorter.
\begin{theorem}
\label{theorem:theorem1} 
Algorithm $\algo$ is an $O(m)$-time $2$-approximation algorithm for
the maximum-leaf spanning tree problem.
\end{theorem}

We follow a similar approach to the one introduced by Solis-Oba, Bonsma, and Lowski~\cite{Solis-Oba98,Solis-ObaBL17}: In each round, the algorithm selects a vertex $u$ on the current tree and expands the tree by adding all neighbors of $u$ that are not already in the tree, directly linking them to $u$. The choice of $u$ is based on the number of its neighbors outside the current tree. Our algorithm expands the tree in a greedy manner, aiming to maximize the number of leaves in each round. If a round cannot increase the number of leaves (i.e., each leaf has at most one neighbor outside the tree), the algorithm looks ahead one round and attempts to maximize the number of leaves increased in two consecutive rounds. When no suitable expansion vertices can be found for two consecutive rounds, the algorithm employs a ``depth-first'' expansion strategy to simplify the analysis of the approximation ratio.

\Xomit{
The general approach of our algorithm is similar to 
the algorithm by Solis-Oba, Bonsma, and Lowski~\cite{Solis-Oba98,Solis-ObaBL17}:
Each round chooses a vertex $u$ on the current tree and expands the tree at  $u$ (i.e., adds all neighbors of $u$ outside the current tree by directly linking them to $u$).
The choice of the vertex $u$ is based on the
number of $u$'s neighbors outside of the current tree. 
Our algorithm grows the current tree in a greedy manner.
Intuitively, each round tries to increase the number of leaves of the current tree 
as much as possible. For an iteration that the number of leaves of the current tree cannot increase (i.e., each leaf of the current tree has at most one neighbor outside the current tree), our algorithm looks ahead one iteration and
tries to maximize
the number of leaves of the current tree to be increased in two consecutive iterations. 
The ``depth-first'' expansion when
 no choice of the expansion vertices for
two consecutive
rounds simplifies the analysis of the approximation ratio.
}

The maximum-leaf spanning tree problem has been extensively studied in
the literature.  Most of the early work focused on finding spanning
trees with many leaves in graphs with minimum degree at least $d$ for
some $d\geq 3$. For such graphs, good lower bounds on the number of
leaves achievable in a spanning tree 
have been
derived
in~\cite{BonsmaZ12,BousquetI0MOSW20,DingJS01,GriggsKS89,KleitmanW91,PayanTX84,Storer81}. There has also
been work on polynomial-time solutions to the problem of determining
if a given graph has a spanning tree with at least $k$ leaves for a
fixed $k$. The first such algorithm, running in $O(n^2)$ time,  was due to Fellows and
Langston~\cite{FellowsL92}. 
Bodlaender~\cite{Bodlaender93} improved the running time to $O(m)$.

Section~\ref{section:section2} describes our algorithm.
Section~\ref{section:section3} analyzes the approximation ratio.
Section~\ref{section:section4} shows a linear-time
implementation of our algorithm.  Section~\ref{section:section5}
concludes the paper.

\section{The algorithm}
\label{section:section2}
All graphs in this paper are undirected and simple, i.\,e., having no
multiple edges and self-loops.  Let $H$ be a graph.  Let $V(H)$
(respectively, $E(H)$) denote the vertex (respectively, edge) set of
$H$.  Let $d_H(u)$ with $u\in V(H)$ denote the number of neighbors
of $u$ in $H$.  
Let $L(H)$ (respectively, $V_2(H)$ and $V_0(H)$)
consist of each vertex $u$ with $d_H(u)=1$ (respectively, $d_H(u)\geq 2$
and $d_H(u)=0$).
A vertex $u\in V(H)$ is a {\em leaf} of $H$ if $u\in L(H)$.

Let $G$ be the input $n$-vertex $m$-edge connected graph.  Assume without
loss of generality that $G$ contains a vertex $a$ with $d_G(a)\geq
2$.  Let $T$ be a tree of $G$.  For a vertex $u$ of $T$, let $V_T(u)$
consist of the vertices $v\in V(G)\setminus V(T)$ with $uv\in E(G)$
and let $E_T(u)$ consist of the edges $uv$ of $G$ with $v\in
V_T(u)$. If $|V_T(u)|=1$, then let $v_T(u)$ denote the vertex in
$V_T(u)$.  
\begin{itemize}
\item Let $W_2(T)$ consist of each vertex $u$ of $T$ with $|V_T(u)|\geq
  2$.

\item Let $W_1(T)$ consist of each vertex $u$ of $T$ with $|V_T(u)|=1$
  and $|V_{T\cup E_T(u)}(v_T(u))|\geq 2$.

\item Let $W_0(T)$ consist of each vertex $u$ of $T$ with $|V_T(u)|=1$
  and $|V_{T\cup E_T(u)}(v_T(u))|\leq 1$.
\end{itemize}

\begin{figure}[t]
\hrule
\medskip
\begin{tabbing}
\qquad\=Algorithm $\algo$:\\
\>\qquad\=Let $T$ be the initial tree of $G$ consisting of a vertex
            $a$ of $G$ with $d_G(a)\geq 2$.\\
  
\>      \>Repeat the following steps until $V(T)=V(G)$:\\
\>      \>\qquad\=If $W_2(T)\ne\varnothing$, then\\
\>      \>      \>\qquad\=let $u$ be an arbitrary vertex in $W_2(T)$\\
\>      \>      \>else if $W_1(T)\ne\varnothing$, then\\
\>      \>      \>      \>let $u$ be an arbitrary vertex in $W_1(T)$\\
\>      \>      \>else\\
\>      \>      \>      \>let $u$ be the vertex in $W_0(T)$ that joins $V(T)$ most recently.\\
\>      \>      \>Expand $T$ at $u$ by letting $T=T\cup E_T(u)$.\\
\>      \>Return $T$.
\end{tabbing}
\hrule
\caption{Our algorithm on a connected graph $G$.}
\label{figure:figure1}
\end{figure}
Our algorithm is as in Figure~\ref{figure:figure1}.  That is,
starting with the initial tree $T=\{a\}$ of $G$, we iteratively expand
$T$ until $V(T)=V(G)$.  Each iteration expands $T$ at a vertex $u$ of
$T$ with $|V_T(u)|\ne 0$ (which has to be a leaf of $T$ except when  $u=a$) by letting
$T=T\cup E_T(u)$. Hence, the tree returned by the algorithm is a
spanning tree of $G$.  The
vertex at which $T$ is expanded is chosen from $W_2(T)$, $W_1(T)$, and
$W_0(T)$ in order.  Thus, if $T$ is to be expanded at a vertex in
$W_i(T)$ with $i\in\{0,1\}$, then each $W_j(T)$ with $i<j\leq 2$ is
empty.  Moreover, the vertices in $W_0(T)$ are chosen in the reversed
order of their joining $V(T)$. As a result, once $T$ is about to be
expanded at a vertex $u\in W_0(T)\cup W_1(T)$, then algorithm grows a
path for $T$ from $u$ by repeatedly letting $T=T\cup \{uv_T(u)\}$ and
$u=v_T(u)$ until $|E_T(u)|\ne 1$.  If the halting condition is
$|E_T(u)|=0$ (respectively, $|E_T(u)|\geq 2$), then the grown path
ends at a leaf (respectively, a vertex with degree at least three)
in the final tree returned by the algorithm.
The chosen order of the vertices at which $T$
is expanded is crucial in analyzing the approximation ratio in the
next section.  An example is shown in Figure~\ref{figure:figure2}(a),
where vertex $i$ is the $i$-th vertex at which $T$ is expanded.  

\begin{figure}
\centerline{\includegraphics[width=12cm]{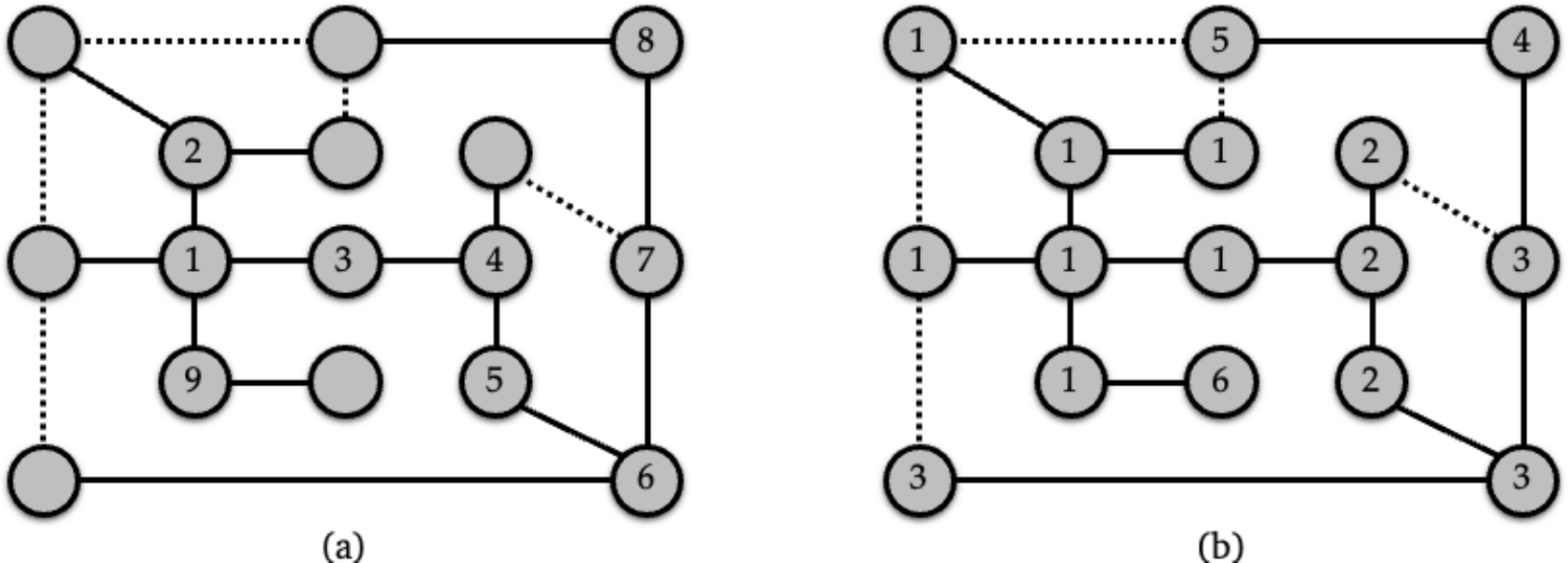}}
\caption{An example of the spanning tree $T$ returned by our
  algorithm. The edges of $T$ are represented by the solid lines. (a)
  Vertex $i$ is the $i$-th expanded vertex. The vertices with no
  labels are the leaves of the returned spanning tree. (b) The label
  of each vertex is its rank according to the expansion order given in
  (a).}
\label{figure:figure2}
\end{figure}

\section{The approximation ratio}
\label{section:section3}
Throughout the section, let $T$ denote the tree returned by our
algorithm rooted at the vertex $a$ of the initial tree.  For each
vertex $v$ of $T$ other than $a$, let $p(v)$ denote the parent of $v$
in $T$.  For each vertex $u$ of $T$ with $d_T(u)\geq 2$, let $T_u$ be
the subtree of $T$ when $T_u$ is about to be expanded at~$u$ and let 
$T^u=T_u\cup E_{T_u}(u)$.
Thus,
$u=p(v)$ if and only if $v\in V_{T_u}(u)$, i.e., $uv$ joins $T$ via
expanding $T_u$ at~$u$.  Define the {\em rank} $r(v)$ of a each $v\in
V(G)$ as follows.  Let $\rank{a}=1$.  For each edge $uv$ of $T$ with 
$u=p(v)$, let
\[
\rank{v}=\left\{
\begin{array}{ll}
\rank{u}&\text{if $u\in W_2(T_u)$}\\
\displaystyle1+\max_{w\in V(T_u)} \rank{w}&\text{otherwise}.
\end{array}
\right.
\]
See Figure~\ref{figure:figure2}(b) for an example.  If $uv$ is an edge of
$G$ with $r(u)<r(v)$, then $v\notin V(T_u)$. That is, when the tree is
about to be expanded at $u$, the tree does not contain any vertex
whose rank is higher than $r(u)$.  Thus, for an edge $uv$ of $G$ with
$r(u)<r(v)$, if $uv\in E(T)$, then $u=p(v)$ and $u\in W_0(T_u)\cup
W_1(T_u)$; otherwise, $u$ is a leaf of $T$ with $r(p(u))=r(u)$.  Let
$U$ consist of the vertices with unique ranks.  The $U$ in  Figure~\ref{figure:figure2}(b) consists of the vertices of ranks $4$, $5$, and $6$. If $u=p(v)$, then the
definition of our algorithm implies
\begin{itemize}
\item
  $u\in W_2(T_u)$ if and only if $\rank{u}=\rank{v}$,
  
\item
  $u\in W_1(T_u)$ if and only if $\rank{u}<\rank{v}$ and $v\notin U$,
  and
  
\item
  $u\in W_0(T_u)$ if and only if $\rank{u}<\rank{v}$ and $v\in U$.
\end{itemize}
\begin{lemma}
\label{lemma:lemma1}
Let $uvw$ is a path of $G$ with
 $\{u,v\}\subseteq U$ and $r(u)<r(v)<r(w)$, then 
$d_G(v)=2$.
\end{lemma}

\begin{proof}
Assume for contradiction that there is a neighbor $x$ of $v$ in $G$ other than $u$ and $w$.
By $\{u,v\}\subseteq U$ and $r(u)<r(v)$, we have $r(u)=r(v)-1$ and
$u\in W_0(T_u)$. Therefore, $W_1(T_u)= W_2(T_u)=\varnothing$, implying
$r(x)<r(u)$. 
By $\{u,w\}\subseteq V_{T'\cup \{xv\}}(v)$, we have
$x\in W_2(T_{p(u)})\cup W_1(T_{p(u)})$, contradicting $p(u)\in W_0(T_{p(u)})$.
\end{proof}

\begin{lemma}
\label{lemma:lemma2}
Each vertex $u$ of $G$ has at most one neighbor $v$ in $G$ with $r(u)<r(v)$.
\end{lemma}

\begin{proof}
Assume for contradiction that $v$ and $w$ are distinct neighbors of
$u$ in $G$ with
\begin{equation}
\label{equation:equation1}
r(u)<r(v)\leq r(w).
\end{equation} 
Thus, $u\ne a$. Let $p=p(u)$.  By Equation~\eqref{equation:equation1},
we have $V(T_p)\cap \{v,w\}=\varnothing$, implying $p\notin
W_0(T_p)$. Also, we have $p\notin W_1(T_p)$; or else $\{v,w\}\subseteq
V_{T_u}(u)$ implies $r(u)=r(v)=r(w)$, violating
Equation~\eqref{equation:equation1}.  Hence, $p\in W_2(T_p)$ and
$\{v,w\}\subseteq V_{T^p}(u)$,
implying $u\in W_2(T^p)$. 
By
Equation~\eqref{equation:equation1}, we have $V_{T_u}(u)\cap
\{v,w\}=\varnothing$. Thus, $u\in W_2(T_u)$, violating
Equation~\eqref{equation:equation1}.
\end{proof}

Let $F$ be the forest obtained from $T$ by deleting all edges $uv$
with $\rank{u}\ne \rank{v}$.  Thus, for any vertices $u$ and $v$ of $G$, we have $\rank{u}=\rank{v}$ if and only if $u$ and $v$ are connected in $F$.

\begin{lemma}
\label{lemma:lemma3}
If $uv$ is an edge of $G$ with $u\in V_2(F)$, then 
$\rank{u}\geq\rank{v}$.
\end{lemma}
\begin{proof}
Assume $r(u)<r(v)$ for contradiction, implying $v\notin V(T_u)$.
By $u\in V_2(F)$, we have $r(u)=r(v)$, contradiction.
\end{proof}

\begin{lemma}
\label{lemma:lemma4}
If $uv$ is an edge of $G$ with 
$u\in U$ and $v\in L(F)$, then
$\rank{u}>\rank{v}$.
\end{lemma}

\begin{proof}
Assume $r(u)\leq r(v)$ for contradiction.
By $u\in U$, we have $r(u)<r(v)$,
implying $v\notin V(T_u)$. By $u\in U$, we have $u=p(v)$, implying $T^u=T_v$.
Thus, $v\in U\cup V_2(F)$, violating $v\in L(F)$.
\end{proof}

Let $F_i$ with $i\geq 1$ be the $i$-th
largest connected component of $F$.  
Let $k$ denote the number of connected components of $F$ having
at least three vertices.
By $d_G(a)\geq 2$, we have $k\geq 1$.
If $uv$ is an edge of $T$ with $r(u)<r(v)$, then $u=p(v)$ and
$d_F(v)\ne 1$.
By definition of our algorithm, $|V(F_i)|\ne 2$ holds for each $i\geq 1$. 
\begin{itemize}
\item
  If $|V(F_i)|=1$, then $F_i$ has at most two incident edges in
  $T$. Thus, $|L(F)|\leq |L(T)|+k-1$.

\item
  If $|V(F_i)|\geq 3$, then $F_i$ has at most one vertex $u$ with
  $d_{F_i}(u)=2$. Thus, $|V(F_i)|\leq 2\cdot |L(F_i)|-1$.
\end{itemize}
We have $U=V_0(F)$
and
\[
n-|U|=\sum_{i=1}^k |V(F_i)| \leq
\sum_{i=1}^k\left(2\cdot |L(F_i)|-1\right) =2\cdot |L(F)|-k \leq
2\cdot |L(T)|+k-2.
\]
To show that $T$ is $2$-approximate, it remains to prove for an
arbitrary spanning tree $R$ of $G$
\begin{equation}
\label{equation:equation2}
|L(R)|\leq n-|U|-k+1.
\end{equation}

\begin{figure}
\centerline{\includegraphics[width=12cm]{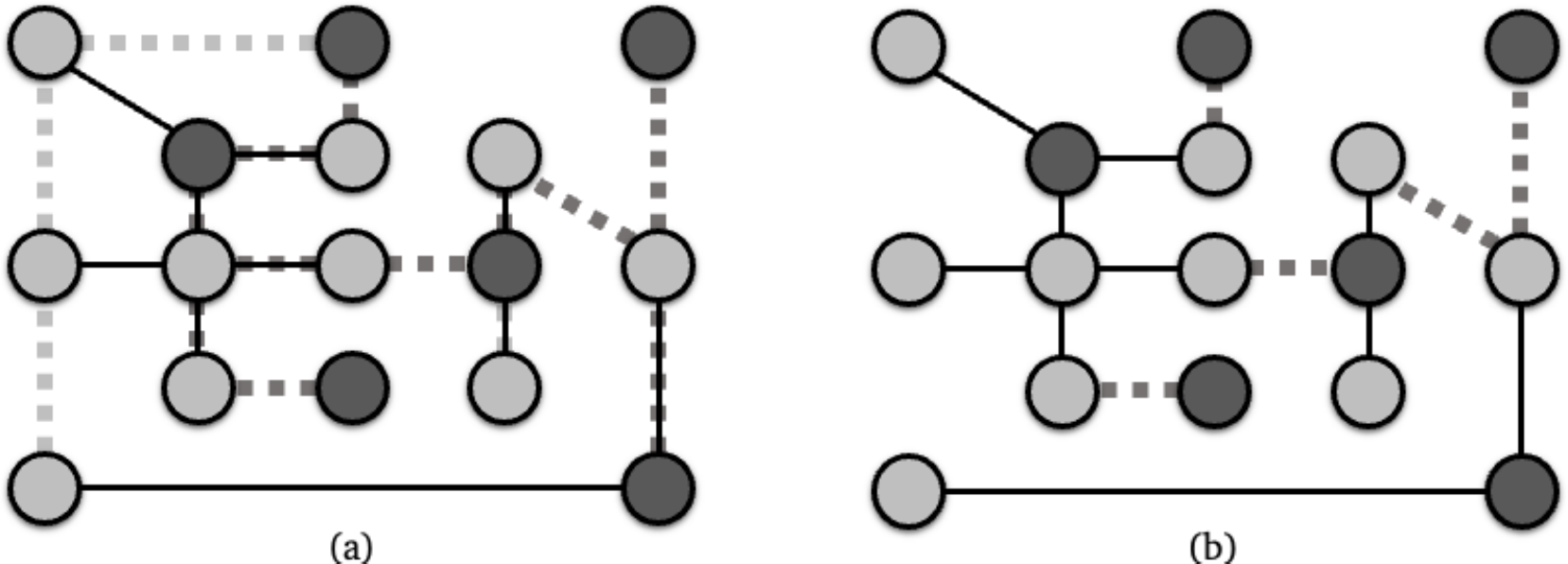}}
\caption{The solid edges form a spanning forest $F$ having $6$ connected components, three of which have at least three vertices.
Any two vertices  are connected in $F$ if and only if their ranks (as shown in Figure~\ref{figure:figure2}(b)) are identical.
The dark vertex in 
each connected component $F_i$ with $1\leq i\leq 6$ is a nonleaf vertex $f_i$ of $F_i$.
(a) The $15$ dotted edges form a spanning tree $R$.
The $11$ dark dotted edges form the minimal subtree $R'$ of $R$ that spans the dark vertices $f_1,\ldots,f_6$. 
(b) The $5$ dotted edges form a minimal forest $S$ with $|V(S)|=9$ of $R'$ such that $F\cup S$ is a spanning tree. $S$ has four connected components.}
\label{figure:figure3}
\end{figure}

Let each $f_i$ with $1\leq i\leq k+|U|$ be an arbitrary vertex of
$V(F_i)\setminus L(F_i)$.  Let $R'$ be the minimum subtree of $R$
that spans all $k+|U|$ vertices $f_i$.  We have $L(R')\cap L(F)=\varnothing$
and 
$L(R)\cap V(R')\subseteq L(R')$, implying 
\begin{equation}
\label{equation:equation3}
L({\color{red}S})\cap L(F)=\varnothing
\end{equation}
for any
minimal forest {\color{red}$S$} of $R'$ such that $S\cup F$ is a spanning tree of $G$.
By minimality of
$S$, 
the vertices of a connected component of $S$ have distinct ranks.
By minimality of $S$, each connected component of $S$ has at least two vertices and
a connected component of $S$ having $j$ vertices connects exactly $j-1$ connected components of $F$. Since $F$ has $k+|U|$ connected components, we have
%
$|V(S)|=|U|+k+t-1$,
where $t\geq 0$ is the number of connected components of $S$.  
See {\color{red}Figure~\ref{figure:figure3}} for an example.
Thus,
\[
|L(R)\setminus V(S)|\leq |V(R)\setminus V(S)|=n-|U|-k-t+1.
\]
We first claim that each connected component of $S$ contains at most
one leaf of $R$, implying $|L(R)\cap V(S)|\leq t$.
Thus, Equation~\eqref{equation:equation2} follows from
\[
|L(R)|=|L(R)\setminus V(S)|+|L(R)\cap V(S)|\leq n-|U|-k+1.
\]
It remains to prove the claim.  
Recall that all vertices of a connected component of $S$ have distinct ranks.
Assume for contradiction that a connected component of $S$ contains 
distinct leaves $u$ and $v$ of $R$ with
$r(u)<r(v)$. By $S\subseteq R'\subseteq R$, we have $\{u,v\}\subseteq L(S)$.
Let $P$ be the $uv$-path of $S$.
For any $\{x,y\}\in V(P)$, let $P[x,y]$ denote the $xy$-path of $P$.  
By Equation~\eqref{equation:equation3},
we have the following cases:

\medskip

\noindent
Case~1: $u\in V_2(F)$. Let $w$ be the neighbor of $u$ in $P$.
By~Lemma~\ref{lemma:lemma3}, 
we have $r(w)<r(u)<r(v)$.
Thus, there is a path $xyz$ of $P$ with $r(y)<r(x)<r(z)$, violating
Lemma~\ref{lemma:lemma2}.

\medskip

\noindent
Case~2: $u\in U$ and $v\in V_2(F)$.  
By Lemmas~\ref{lemma:lemma2} and~\ref{lemma:lemma3}, 
the ranks of the vertices in $P$ are monotonically increasing from $u$ to $v$.
Let $xy$ be the edge of $P$ with $V(P[u,x])\subseteq U$ and $y\notin U$.
By $x\in U$, we have $p(x)\in W_0(T_{p(x)})$.
Thus, $T^{p(x)}=T_x$, implying
$x=p(y)$ by $r(x)<r(y)$.
By $x=p(y)$ and $y\notin U$, we have $y\in V_2(F)$.
We have $y=v$; or else the rank of the neighbor of $y$ in $P[y,v]$ is
lower than $r(y)$ by Lemma~\ref{lemma:lemma3}. 
Lemma~\ref{lemma:lemma1} implies
$d_R(w)=2$ for each interior vertex $w$ of $P$. Thus, $a$ and $u$ are not connected in $R$, contradiction.

\medskip
\noindent
Case~3: $\{u,v\}\subseteq U$. 
\begin{itemize}
\item Case~3a:
A set $\{x,z\}=\{u,v\}$ and an interior vertex $y$ of $P$ with $y\in L(F)$ 
satisfy $V(P[y,z]-y)\subseteq U$.
Since the neighbor of $y$ in $P[y,z]$ is in $U$, 
its rank is higher than $r(y)$ by Lemma~\ref{lemma:lemma4}.
By Lemma~\ref{lemma:lemma2}, the ranks of the vertices of $P[x,y]$
are monotonically increasing from $x$ to $y$.
Thus, $V(P[x,y])\cap V_2(F)=\varnothing$
by Lemma~\ref{lemma:lemma3}.
By $x\in U$ and $y\in L(F)$, 
there is a
vertex $w$ of $P[x,y]$ in $L(F)$ such that it neighbor in $P[x,w]$ is in $U$, contradicting Lemma~\ref{lemma:lemma4}.

\item Case~3b:  Case~3a does not hold and there is an interior vertex $w$ of $P$ in $V_2(F)$.
By Lemmas~\ref{lemma:lemma2} and~\ref{lemma:lemma3},
the ranks of the vertices of $P[x,w]$ are monotonically increasing 
for each $x\in\{u,v\}$, implying that $w$ is the only vertex of $P$ in $V_2(F)$.
Since Case~3a does not hold, we have $V(P-w)\subseteq U$.
Let $y$ and $z$ be the neighbors of $w$ in $P$ with $r(y)<r(z)$.
Thus, $V(T_y)\cap \{w,z\}=\varnothing$.
By $y\in U$ and $V_{T^{p(y)}}(y)\ne\varnothing$, 
we have $T^{p(y)}=T_y$. Lemma~\ref{lemma:lemma1} implies
$d_G(y)=2$.
Thus, $y=p(w)$ and $T^y=T_w$, implying $r(w)=r(z)$, contradiction.

\item Case~3c: Neither of Cases~3a and~3b holds. Thus, $V(P)\subseteq U$.
Since $a$ and $u$ are connected in $R$ and 
$\{u,v\}\subseteq L(R)$,
there is an interior vertex $w$ of $P$ with $d_G(w)\geq 3$.
By Lemmas~\ref{lemma:lemma1} and \ref{lemma:lemma2}, the
ranks of the vertices of $P$ are monotonically increasing from $u$ (respectively, $v$) to $w$.
Let $x$ and $y$ be the neighbors of $w$ in $P$ with $r(x)<r(y)$. 
Let $z$ be a neighbor of $w$ in $G$ other than $x$ and $y$.
\begin{itemize}
\item If $r(x)<r(z)$, then $V(T_x)\cap\{y,z,w\}=\varnothing$. By $x\in U$ and $V_{T^{p(x)}}(x)\ne\varnothing$, we have $T^{p(x)}=T_x$. Lemma~\ref{lemma:lemma1} implies $d_G(x)=2$. Thus, $x=p(w)$ and $T^x=T_w$, implying $r(w)=r(y)$, contradiction.

\item If $r(z)<r(x)$, then $V(T_z)\cap \{x,y,w\}=\varnothing$. By $w\in U$, we have
$T^z=T_w$, implying $r(x)=r(w)=r(y)$, contradiction.
\end{itemize}
\end{itemize}

\begin{figure}
\centerline{\includegraphics[width=6cm]{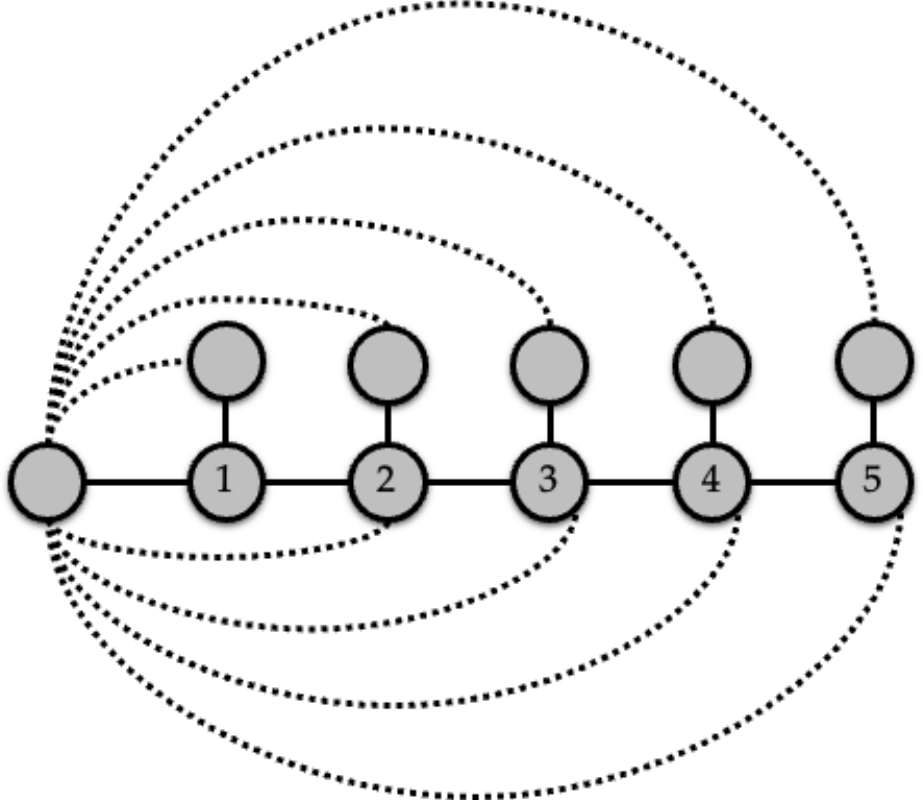}}
\caption{The spanning tree $T$ consists of the solid edges. The number of
each internal vertex of $T$ represents its expansion order.}
\label{figure:figure4}
\end{figure}

Our analysis is tight. Consider the graph $G$ in
Figure~\ref{figure:figure4} composed of the solid and the dashed lines. If
vertex $i$ is the $i$-vertex expanded, then the spanning tree $T$
computed by our algorithm is exactly the one consisting of the solid
lines with $|\leaf{T}|=6$.  However, $G$ has a spanning tree $R$
rooted at the highest-degree vertex with $|\leaf{R}|=10$, which is $2|\leaf{T}|-2$.

\section{A linear-time implementation}
\label{section:section4}
Our algorithm can easily be implemented to run in $O(m)$ time using
merely arrays and linked lists.  We maintain three (``waiting'') lists
$W_2$, $W_1$, and $W_0$ of the leaves of $T$ and a list $U_w$ for each
$w\in V(G)$ of the (``unspanned'') neighbors of $w$ in $V(G)\setminus
V(T)$.  To expand $T$ at a vertex $u\in V(T)$, we (i) insert each
vertex $v\in V_T(u)$ to the bottom of $W_2$ and let $U_w=U_w\setminus
\{v\}$ for each $vw\in E(G)$ and (ii) let $T=T\cup E_T(u)$.  Thus, it
takes overall $O(m)$ time to expand $T$ throughout the algorithm.  The
algorithm expands the initial $T=\{a\}$ at $a$ and enters the loop,
each of whose iterations runs the following instructions.  Note that,
according to the above implementation, 
expanding $T$ at a vertex $u\in V(T)$ with $|V_T(u)|=0$ does nothing.

\begin{itemize}
\item Case~1: $W_2\ne \varnothing$. Delete a vertex $u$ from the top
  of $W_2$.  If $|V_T(u)|=1$, then insert $u$ to the bottom of $W_1$.
  Otherwise, expand $T$ at $u$.

\item Case~2: $W_2=\varnothing$ and $W_1\ne\varnothing$. Delete a
  vertex $u$ from the top of $W_1$. If $|V_{T\cup E_T(u)}(v_T(u))|=1$,
  then insert $u$ to the bottom of $W_0$. Otherwise, expand $T$ at
  $u$.

\item Case~3: $W_2=W_1=\varnothing$ and $W_0\ne\varnothing$.  Delete a
  vertex $u$ from the bottom of $W_2$.  Expand $T$ at $u$.
\end{itemize}
Each condition involving $|V_T(u)|$ or $|V_{T\cup E_T(u)}(v_T(u))|$
can be determined in $O(1)$ time using the lists $U_w$.  Thus,
excluding the time for expanding $T$, the loop takes overall $O(n)$
time.  The three lists $W_2$, $W_1$, and $W_0$ enforce that the
vertex 
$u$ 
of $T$ are chosen in the order of $W_2(T)$, $W_1(T)$, and
$W_0(T)$.  Since each $W_i$ with $i\in\{0,1,2\}$ preserves the order
of its vertices that join $V(T)$ and and $W_0$ is processed in the
reversed order, the vertices in $W_0(T)$ are chosen in the reversed
order of their joining $V(T)$.  Hence, our algorithm is correctly
implemented above to run in $O(m)$ time.
Theorem~\ref{theorem:theorem1} is proved.

\section{Conclusion}
\label{section:section5}
We present a simple linear-time 2-approximation algorithm for the
problem of maximum-leaf spanning tree.  Our analysis for the
approximation ratio is tight.  However, as suggested by the example
shown in Figure~\ref{figure:figure4}, it is of interest to see if the
approximation ratio can be improved by starting our algorithm with a
vertex having the highest degree. The problem is MAX~SNP-complete, so
it is also of interest to have nontrivial lower bounds on the
approximation ratio.

\bibliographystyle{hilabbrv}
\bibliography{parking}
 
\end{CJK*}
\end{document}